\documentclass{jpp}

\newcommand{\be}{\begin{equation}}
\newcommand{\ee}{\end{equation}}
\let\proof\relax 
\usepackage[psamsfonts]{amsfonts}
\usepackage{amsmath, amsthm, amssymb}
\usepackage{multirow}
\usepackage{enumerate}
\usepackage{setspace}
\usepackage[dvips]{graphicx}
\usepackage{natbib}

\newcommand{\derv}[1]{\frac{\partial}{\partial #1}}

\newcommand{\deriv}[2]{\frac{\partial #1}{\partial #2}}

\newcommand{\beqn}{\begin{equation}}
\newcommand{\eeqn}{\end{equation}}
\newcommand{\beqnar}{\begin{eqnarray}}
\newcommand{\eeqnar}{\end{eqnarray}}

\newtheorem{theorem}{Theorem}[section]
\newtheorem{proposition}[theorem]{Proposition}

\newenvironment{remark}[1][Remark]{\begin{trivlist}
\item[\hskip \labelsep {\bfseries #1}]}{\end{trivlist}}

\ifCUPmtlplainloaded \else
  \checkfont{eurm10}
  \iffontfound
    \IfFileExists{upmath.sty}
      {\typeout{^^JFound AMS Euler Roman fonts on the system,
                   using the 'upmath' package.^^J}%
       \usepackage{upmath}}
      {\typeout{^^JFound AMS Euler Roman fonts on the system, but you
                   dont seem to have the}%
       \typeout{'upmath' package installed. JPP.cls can take advantage
                 of these fonts, if you use 'upmath' package.^^J}%
       \providecommand\upi{\pi}%
      }
  \else
    \providecommand\upi{\pi}%
  \fi
\fi

\ifCUPmtlplainloaded \else
  \checkfont{msam10}
  \iffontfound
    \IfFileExists{amssymb.sty}
      {\typeout{^^JFound AMS Symbol fonts on the system, using the
                'amssymb' package.^^J}%
       \usepackage{amssymb}%
       \let\le=\leqslant  
       \let\ge=\geqslant  
      }{}
  \fi
\fi

\ifCUPmtlplainloaded \else
  \IfFileExists{amsbsy.sty}
    {\typeout{^^JFound the 'amsbsy' package on the system, using it.^^J}%
     \usepackage{amsbsy}}
    {\providecommand\boldsymbol[1]{\mbox{\boldmath $##1$}}}
\fi

\providecommand\bnabla{\boldsymbol{\nabla}}
\providecommand\bcdot{\boldsymbol{\cdot}}

\newcommand\etb{\boldsymbol{\eta}}

\newcommand\Real{\mbox{Re}} 
\newcommand\Imag{\mbox{Im}} 
\newcommand\Ai{\mbox{Ai}}            
\newcommand\Bi{\mbox{Bi}}            

\newcommand\slsQ{\mathsfbi{Q}} 

\newsavebox{\astrutbox}
\sbox{\astrutbox}{\rule[-5pt]{0pt}{20pt}}
\newcommand{\astrut}{\usebox{\astrutbox}}

\newcommand\p{\ensuremath{\partial}}

\newcommand\kgd{\ensuremath{k\gamma d}}

\newcommand\squart{\ensuremath{{\textstyle\frac{1}{4}}}}
\newcommand\thalf{\ensuremath{{\textstyle\frac{1}{2}}}}
\newcommand\Gat{\ensuremath{\widetilde{G_a}}}
\newcommand\ttz{\ensuremath{\rightarrow 0}}
\newcommand\ndq{\ensuremath{\frac{\mbox{$\partial$}}{\mbox{$\partial$} n_q}}}
\newcommand\sumjm{\ensuremath{\sum_{j=1}^{M}}}
\newcommand\pvi{\ensuremath{\int_0^{\infty}%
  \mskip \ifCUPmtlplainloaded -30mu\else -33mu\fi -\quad}}

\newcommand\etal{\mbox{\textit{et al.}}}

\newtheorem{lemma}{Lemma}
\newtheorem{corollary}{Corollary}

\title[Ion Acoustic Travelling Waves in Multi-Fluid Plasmas]
{Ion Acoustic Travelling Waves}

\author[G. M. Webb, R. H. Burrows, X. Ao  and G. P. Zank]%
{G.\ns M.\ns W\ls E\ls B\ls B$^1$%
 \thanks{Email address for correspondence: gmw0002@uah.edu},\ns
R.\ls H.\ns B\ls U\ls R\ls R\ls O\ls W\ls S$^1$\break
X.\ns A\ls O$^1$
 \and G.\ns P.\ns Z\ls A\ls N\ls K$^{1,2}$}

\affiliation{$^1$Center for Space Plasma and Aeronomic Research, 
The University of Alabama in Huntsville, 
Huntsville AL 35805, USA\\[\affilskip]
$^2$Department of Physics, The University of
Alabama in Huntsville, Huntsville AL 35899, USA}

\begin{document}


\maketitle

\begin{abstract}
Models for travelling waves in multi-fluid plasmas 
give essential insight 
into fully nonlinear wave structures in plasmas, not readily available 
from either 
numerical simulations or from weakly nonlinear wave theories. We illustrate
these ideas using one of the simplest models of an electron-proton 
multi-fluid plasma for the case where there is no magnetic field 
or a constant normal magnetic field present. We show that the travelling 
waves can be reduced to a single first order differential equation 
governing the dynamics. We also show that the equations admit a 
multi-symplectic Hamiltonian formulation in which both the space and 
time variables can act as the evolution variable.  An integral equation 
useful for calculating adiabatic, electrostatic solitary wave signatures 
for multi-fluid plasmas with arbitrary mass ratios is presented. The 
integral equation arises naturally from a fluid dynamics approach 
for  a two fluid plasma, with a given mass ratio of the two 
species (e.g. the plasma could be an electron proton or an electron 
positron plasma). Besides its intrinsic interest, the integral 
equation solution provides a useful analytical test for numerical codes 
that include a proton-electron mass ratio as a fundamental constant, such as 
for particle in cell (PIC) codes.  
The integral equation is used to delineate the physical characteristics of 
ion acoustic 
travelling waves consisting of  hot electron and cold proton fluids. 
\end{abstract}




\section{Introduction}
Electrostatic solitary waves (ESW's) have been observed simultaneously with 
reflected suprathermal ions at collisionless shocks, 
e.g. by the geotail spacecraft
 at Earth's bow shock (\citet{Shin08}). 
Such waves are observed to be of the order 
of 2-7 Debye lengths and amplitudes of 
$\sim$ 100 mV/m, by the WIND spacecraft.  

Quasi-perpendicular shock models in which suprathermal ions gain energy 
in the motional electric field upon reflection from an electrostatic shock 
potential (ESSP) have had some success as a possible dissipation mechanism 
for super-critical collisionless shocks, and as an explanation for some 
observations. For example, reflected pick-up ions can explain the cooler 
than expected solar wind observed by Voyager 2 downstream of the heliospheric 
termination shock [\cite{Richardson08},\cite{Zank96,Zank10}, 
\cite{Burrows10},\cite{Oka11}]. \cite{Zank96,Zank10} 
point out the importance of the cross shock electrostatic potential at the 
solar wind termination shock and at travelling interplanetary shocks, 
and the acceleration of pick-up ions by the shock surfing mechanism. The 
dissipation mechanism for the solar wind termination shock 
is  due in large part to the interaction of pick-up ions 
with the shock since they
 carry  most of the  momentum flux of the 
solar wind (\cite{Zank96,Zank10}). \cite{Oka11} studied in detail the 
dissipation mechanism and pick-up ion distribution at the Solar Wind 
termination shock by using PIC simulations.

Observations indicate that electrostatic solitary waves reflect 
particles and are fundamental components of shocks
 \citep{Shin08,Zank96,Wilson07}. Self consistent models of the solar wind 
termination shock based in part on a fluid dynamics approach 
needs to incorporate the effects of reflected particles such as 
shock surfing pick-up ions. These models cannot make the typical assumption of
charge neutrality since electrostatic structures arise from charge separation 
(e.g. \cite{Tidman71}). To show the relevance of our solutions to simulations
 that use an artificial ratio of the electron to proton mass 
(e.g. \cite{Oka10,Oka11}) we discuss the electron to proton mass ratio 
dependence of the amplitude of ion-acoustic solitary waves, and find that 
the amplitude of the wave increases with increasing $m_e/m_p$. 
We also find that 
the wave amplitude increases with the ion acoustic Mach number, and the width of the solitons decreases with increasing Mach number, which may be important 
in understanding the dissipation mechanism at quasi-perpendicular shocks 
(\cite{Zank96,Burrows10,Oka10,Lipatov99}).  
  
There are more complicated models of ion acoustic waves than the model 
adopted in the present paper (see e.g. \cite{Baluku10} 
who investigated ion acoustic solitary waves in a plasma with both cool and hot 
electrons). We will not use these more complicated models in the present paper. 

 McKenzie and co-workers investigated a variety 
of two fluid models of fully nonlinear travelling waves in space plasmas 
(e.g. \cite{McKenzie-etal04}
).  These solutions encompass both cases in 
which the charge neutrality assumption is a valid approximation 
(e.g. for whistler oscillitons), and also  for other cases where the 
charge neutrality assumption is not a good approximation (e.g. for 
ion-acoustic traveling waves where the charge separation electric field 
is essential in describing the wave structure).  
\cite{Verheest04a} 
discuss the charge neutrality assumption for 
whistler oscillitons. \cite{Dubinin07}  carried out extensive data analysis 
on the coherent whistler emissions in the magneto-sphere Cluster observations, 
and found strong evidence of whistler oscillitons in the data. 
\cite{Dubinov07a,Dubinov07b}  studied periodic space charge waves. In the limit 
as the wave spatial period goes to infinity, these solutions reduce to the 
ion-acoustic solitons studied by \cite{McKenzie02}.  
\cite{Dubinov07a,Dubinov07b} did not use the same parameters as 
\cite{McKenzie02},
but the underlying model is the same as that used by 
\cite{McKenzie02}. 
\cite{McKenzie02} did not consider periodic travelling waves with a 
finite wave period.

\cite{Webb05,Webb07,Webb08}  developed a Hamiltonian formulation  
for nonlinear travelling whistler waves in quasi-charge neutral plasmas. 
 \cite{Sauer01,Sauer02,Sauer03}, 
\cite{Dubinin03, Dubinin07} and   
\cite{McKenzie-etal04}  studied  whistler oscillitons.  
\cite{Webb07, Webb08} showed that the travelling waves in this model 
are described by two different but equivalent Hamiltonian formulations.
In the first formulation, the Hamiltonian is the total 
conserved, longitudinal $x$-momentum integral of the system, in 
which the energy flux integral $\varepsilon=const.$ is a constraint 
and for which $d/dx$ is the Hamiltonian evolution operator. 
In the second Hamiltonian approach, the Hamiltonian is the energy flux 
integral $\varepsilon$, in which the $x$-momentum integral $P_x=const.$ 
is a constraint. In the latter case, the Hamiltonian evolution operator 
is the advective Lagrangian time derivative operator $d/d\tau=u_x d/dx$. 
These dual variational principles are analogous to the dual or multi-symplectic variational principles obtained by \cite{Bridges92} in studies of travelling 
water waves (see also \cite{Bridges97a,Bridges97b,Bridges05}). 
\cite{McKenzie-etal06}  
cast the spatial evolution
equations for solitary travelling waves in a Hall current plasma  
in Hamiltonian form, in which the energy flux integral $\varepsilon$ 
is the Hamiltonian and the longitudinal momentum flux integral $P_x=const.$
 acts as a constraint. \cite{Mace07} derived conservation laws for 
travelling waves in multi-fluid plasmas using Bernoulli type theorems 
and generalized vorticities for the different plasma species. 
 \cite{Hydon05} develops the general theory for 
multi-symplectic Hamiltonian systems, and \cite{Cotter07} show that 
multi-symplectic equations can be derived for fluid systems based on Clebsch 
variables which act as canonically conjugate momenta, and in which the 
Clebsch variables are Lagrange multipliers for the constraints in the 
variational principle. 
 \cite{Bridges01,Bridges06}, \cite{Marsden98} and \cite{Marsden99}  
 use multi-symplectic systems in numerical finite difference
schemes  for Hamiltonian wave equations (see also 
\cite{Brio10} for a discussion of multi-symplectic difference schemes). 

In the present paper we develop a multi-symplectic 
description for ion acoustic travelling waves in an electron proton, 
or electron positron plasma, in which there is a non-trivial electric field
induced by charge separation.  We also present a detailed description of ion 
acoustic travelling waves in order to illustrate the different 
types of solution that are possible (see \cite{McKenzie02,McKenzie03a}, 
\cite{McKenzie-Doyle03}, \cite{McKenzie-etal04} 
for a related 
analysis of fully nonlinear, ion acoustic  travelling waves). 
In Section 2, we present the basic equations of the model. 
 In order to keep 
the analysis simple we consider only the case of plasmas in which there is 
no magnetic field present, or a constant parallel magnetic field is present in 
the travelling wave. In Section 3, we discuss the dispersion equation 
for linear waves in a two-fluid, electron-ion  plasma. 
In Section 4, we develop the dual variational 
principles for travelling, ion acoustic waves. 
Section 5 presents examples of the travelling ion acoustic waves, 
illustrating the physics and including a discussion of their possible 
applications to both spacecraft observations, and simulations. 
Section 6 concludes with a summary and discussion.

\section{Basic Equations}
In this section we formulate the travelling wave solution in an electron 
proton two-fluid model as a Hamiltonian system. The analysis is based on 
the multi-fluid equations for an electron proton fluid. To simplify the 
analysis we consider only the case where there is a charge separation 
electric field, but there is no magnetic field. The  electron and proton 
mass continuity equations for the system are:
\begin{align}
&\deriv{n_e}{t}+\nabla{\bf\cdot}(n_e {\bf u}_e)=0, \nonumber\\
&\deriv{n_p}{t}+\nabla{\bf\cdot}(n_p {\bf u}_p)=0, \label{eq:ham1}
\end{align}
where $n_e$ and $n_p$ are the electron and proton fluid number densities, 
and ${\bf u}_e$ and ${\bf u}_p$ are the electron and proton fluid velocities.
Poisson's equation for the system is:
\begin{equation}
\varepsilon_0 \nabla{\bf\cdot}{\bf E}=e(n_p-n_e). \label{eq:ham2}
\end{equation}
where ${\bf E}$ is the electric field. 
The momentum equations for the system can be written in the form:
\begin{align}
&\deriv{{\bf u}_e}{t}+{\bf u}_e{\bf\cdot}\nabla {\bf u}_e
=-\frac{1}{m_en_e}\nabla p_e-\frac{e}{m_e} {\bf E}, \nonumber\\
&\deriv{{\bf u}_p}{t}+{\bf u}_p{\bf\cdot}\nabla {\bf u}_p
=-\frac{1}{m_pn_p}\nabla p_p+\frac{e}{m_p} {\bf E}. \label{eq:ham3}
\end{align}
To complete the equation system, we assume a polytropic equation of state
for the electron and proton fluids, i.e. $p_a=p_{a0}(s_a)n_a^{\gamma_a}$ 
 where $s_a$ is the entropy ($a=e,p$). For simplicity, we assume 
that the entropies $s_a$ are constants throughout the flow (this restriction 
can be lifted if necessary). 

Poisson's equation (\ref{eq:ham2}) is related to the charge conservation 
law:
\begin{equation}
\deriv{{\rho_q}}{t}+\nabla{\bf\cdot}{\bf J}=0, \label{eq:ham4}
\end{equation}
where
\begin{equation}
\rho_q=e(n_p-n_e),\quad {\bf J}=e(n_p{\bf u}_p-n_e {\bf u}_e), \label{eq:ham5}
\end{equation}
define the charge density  $\rho_q$  and the electric current 
density ${\bf J}$. The electric current ${\bf J}$ is related  
 Ampere's law:
\begin{equation}
\nabla\times {\bf H}={\bf J}+\deriv{\bf D}{t}, \label{eq:ham6}
\end{equation} 
where ${\bf H}={\bf B}/\mu_0$ is the magnetic field strength, ${\bf B}$ is the 
magnetic field induction, and ${\bf D}=\varepsilon_0{\bf E}$  is the electric 
field displacement. Taking the divergence of Ampere's 
equation (\ref{eq:ham6}) gives the charge conservation law (\ref{eq:ham4}) 
where we identify the charge density $\rho_q=\nabla{\bf\cdot}{\bf D}$ with the 
divergence of the electric field displacement ${\bf D}$. 
The latter equation 
is equivalent to Poisson's equation (\ref{eq:ham2}).

\section{Dispersion Equation}
For the electron-ion, two fluid plasma model (\ref{eq:ham1})-(\ref{eq:ham5})
without magnetic field effects, the perturbed equations governing linear 
waves have the form:
\begin{align}
&\deriv{\delta n_i}{t}+\derv{x}(n_0 \delta u_i)=0, \nonumber\\
&\deriv{\delta n_e}{t}+\derv{x}(n_0 \delta u_e)=0, \nonumber\\
&n_0 m_i\deriv{\delta u_i}{t}=e n_0\delta E
-\frac{\gamma_i p_i}{n_0}\deriv{\delta n_i}{x}, \nonumber\\
&n_0 m_e \deriv{\delta u_e}{t}=-en_0\delta E
-\frac{\gamma_e p_e}{n_0} \deriv{\delta n_e}{x},\nonumber\\
&\deriv{\delta E}{x}=4\pi e \left(\delta n_i-\delta n_e\right), 
\label{eq:disp1}
\end{align}
where $\delta\psi$ denotes the perturbation of the physical quantity $\psi$. 
Assuming perturbations of the form $\delta\psi\propto \exp(ik x-i\omega t)$, 
the system (\ref{eq:disp1}) reduces to the algebraic equation:
\begin{equation}
k^2\left(1-\frac{\omega_{pi}^2}{\omega^2-k^2 c_i^2} 
-\frac{\omega_{pe}^2}{\omega^2-k^2 c_e^2}\right) \delta E=0. 
\label{eq:disp2}
\end{equation}
Thus, the dispersion equation for the system is:
\begin{equation}
D(k,\omega)=1-\frac{\omega_{pi}^2}{\omega^2-k^2 c_i^2}
-\frac{\omega_{pe}^2}{\omega^2-k^2 c_e^2}=0, \label{eq:disp3}
\end{equation}
where
\begin{equation}
c_i=\left(\frac{\gamma_i p_i}{n_0 m_i}\right)^{1/2},\quad 
c_e=\left(\frac{\gamma_e p_e}{n_0 m_e}\right)^{1/2}, \label{eq:disp4}
\end{equation}
are the adiabatic sound speeds for the ions ($c_i$) and electrons ($c_e$). 
In (\ref{eq:disp1}) we assume adiabatic equations of state for the 
electron and proton fluids, with adiabatic indices $\gamma_i$ and $\gamma_e$.
The quantities 
\begin{equation}
\omega_{pi}=\left(\frac{4\pi n e^2}{m_i}\right)^{1/2}\quad \hbox{and}\quad 
\omega_{pe}=\left(\frac{4\pi n e^2}{m_e}\right)^{1/2}, \label{eq:disp5}
\end{equation}
are the ion and electron plasma frequencies respectively. 

The general structure of the dispersion equation (\ref{eq:disp3}) 
may be deduced by writing (\ref{eq:disp3}) in the form:
\begin{equation}
k^2=\frac{\omega_{pi}^2}{\lambda^2-c_i^2} 
+\frac{\omega_{pe}^2}{\lambda^2-c_e^2}\quad\hbox{where}\quad 
\lambda=\frac{\omega}{k},  \label{eq:disp6}
\end{equation}
is the phase speed. A sketch of $k^2$ versus $\lambda$ (assuming $c_e>c_i$) 
reveals that the roots for the phase speed $\lambda=\omega/k$ 
are located in the ranges (a)\ $c_i^2<\lambda^2<\zeta^2$ and 
(b) $\lambda^2>c_e^2$ where
\begin{equation}
\zeta^2=\frac{c_i^2\omega_{pe}^2+c_e^2\omega_{pi}^2}{\omega_{pe}^2
+\omega_{pi}^2}.  \label{eq:disp7}
\end{equation}
The solutions (a) correspond to the ion acoustic branch of the dispersion 
equation, and range (b) corresponds to the Langmuir wave branch. The 
quantity $\zeta$ is the generalized ion acoustic speed. As $k\to\infty$, 
$\lambda\to c_i$ for branch (a) and $\lambda\to c_e$ for branch (b), 
i.e. the characteristic wave speeds at short length scales ($k\to\infty$) 
correspond to electron and ion sound waves. 

Alternatively, $\omega^2$ from (\ref{eq:disp3}) satisfies 
the biquadratic equation:
\begin{equation}
\omega^4-\alpha \omega^2+\beta=0, \label{eq:disp8}
\end{equation}
where
\begin{align}
&\alpha=k^2 c_e^2+k^2 c_i^2+\omega_{pi}^2+\omega_{pe}^2, \nonumber\\
&\beta=k^4 c_i^2c_e^2+k^2(\omega_{pi}^2 c_e^2+\omega_{pe}^2 c_i^2). 
\label{eq:disp9}
\end{align} 
The solutions of (\ref{eq:disp8}) are
\begin{equation}
\omega_\pm^2=\frac{1}{2}\left(\alpha\pm \sqrt{\Delta}\right), 
\quad \Delta=\alpha^2-4\beta.  \label{eq:disp10}
\end{equation}
The discriminant $\Delta$ reduces to:
\begin{equation}
\Delta= \left(\omega_{pe}^2+\omega_{pi}^2\right)^2
+2 k^2 (\omega_{pe}^2-\omega_{pi}^2)(c_e^2-c_i^2)+k^4 (c_e^2-c_i^2)^2. 
\label{eq:disp11}
\end{equation}
The $\omega_+^2$ root corresponds to the Langmuir wave branch (b) and 
$\omega_-^2$ corresponds to the ion-acoustic wave branch (a). 

\subsection{Long wavelength approximation}
It is instructive to investigate the dispersion equation (\ref{eq:disp8}) 
solutions in the limit as $k\to 0$. As $k\to 0$  the approximate solution 
of the dispersion equation for the ion acoustic wave to $O(k^4)$  is:
\begin{equation}
\omega=k c_{ia}-\gamma k^3, \label{eq:disp12}
\end{equation}
where
\begin{align}
&c_{ia}\equiv \zeta=\left(\frac{c_i^2 \omega_{pe}^2
+c_e^2\omega_{pi}^2}{\omega_{pe}^2+\omega_{pi}^2}\right)^{1/2}, 
\label{eq:disp13}\\
&\gamma=\frac{\omega_{pe}^2\omega_{pi}^2}{\left(\omega_{pe}^2
+\omega_{pi}^2\right)^3} \frac{(c_e^2-c_i^2)^2}{2 c_{ia}}. 
\label{eq:disp14}
\end{align}
Here $c_{ia}$ is the generalized ion-acoustic wave speed 
(in the limit $\omega_{pe}^2>>\omega_{pi}^2$ and $c_i=0$ the ion 
acoustic speed $c_{ia}=[\gamma_e p_e/(n m_i)]^{1/2}$), and 
$\gamma$ gives the dispersion of the wave. Using the association 
$\partial/\partial t\to -i\omega$ and $\partial/\partial x\to ik$ 
for perturbations $\delta\psi\propto \exp(ikx-i\omega t)$, (\ref{eq:disp12}) 
gives the linearized Korteweg de Vries (KdV) equation:
\begin{equation}
\left(\derv{t}+c_{ia}\derv{x}+\gamma \frac{\partial^3}{\partial x^3}\right) 
\delta\psi=0, \label{eq:disp15}
\end{equation}
for long-wavelength ion acoustic waves. The original derivation of the KdV 
equation describing ion acoustic solitons was derived by  
\cite{Washimi66}. 
Their derivation assumed isothermal hot electrons and cold ions. 
\cite{Moslem99,Moslem00}  derived the KdV equation for ion acoustic waves, 
including negative ions, positrons and other particle species. They assumed 
adiabatic equations of state for the different species, with adiabatic 
index of $\gamma=3$, and included a beam plasma component. 
 Other authors developed kinetic plasma models for the 
ion acoustic wave including the effects of Landau damping. In particular, 
 \cite{Ott69}  developed a nonlinear theory for a nonlinear ion 
acoustic wave. They obtained  
 a KdV equation, with a Cauchy Principal value integral 
term representing the effects of Landau damping.   

Dimensional analysis of (\ref{eq:disp12})-(\ref{eq:disp15}) gives the 
characteristic length and time scales:
\begin{equation}
L=\sqrt{\frac{\gamma}{c_{ia}}}, \quad 
T=\frac{L}{c_{ia}}, \label{eq:disp16}
\end{equation}
for the waves. In the cold ion gas limit ($c_i=0$) , and for an electron-proton plasma with $\omega_{pe}^2>>\omega_{pi}^2$, the length and time scales 
(\ref{eq:disp16}) reduce to:
\begin{equation}
L=\frac{c_e}{\omega_{pe}},\quad T=\frac{1}{\omega_{pi}}\quad\hbox{and}\quad 
c_{ia}\approx \left(\frac{\gamma_e p_e}{n m_i}\right)^{1/2}. 
\label{eq:disp17}
\end{equation}
It is interesting to note that $\gamma=0$ if $c_e=c_i$. In this case  
the ion acoustic wave is non-dispersive (see subsection 3.1.1; note however, 
for a hot plasma Landau damping is important and a plasma kinetic treatment 
is required). 

A similar analysis of the Langmuir wave branch  
 in (\ref{eq:disp8})-(\ref{eq:disp10}) for $\omega_+^2$ gives:
\begin{equation}
\omega_+^2=\omega_{pe}^2+\omega_{pi}^2 +k^2 a_2+k^4 a_4, \label{eq:disp18}
\end{equation}
where
\begin{align}
&a_2=\frac{c_e^2\omega_{pe}^2+c_i^2\omega_{pi}^2}{\omega_{pe}^2+\omega_{pi}^2}, 
\label{eq:disp19}\\
&a_4=\frac{\omega_{pe}^2\omega_{pi}^2(c_e^2-c_i^2)^2}
{(\omega_{pe}^2+\omega_{pi}^2)^3}. \label{eq:disp20}
\end{align}
Using the association $\partial_t=-i\omega$ and $\partial_x=ik$, and dropping 
the $O(k^4)$ term in (\ref{eq:disp18}) we obtain the Klein-Gordon equation:
\begin{equation}
\left(\frac{\partial^2}{\partial t^2}-a_2\frac{\partial^2}{\partial x^2} 
+\left(\omega_{pe}^2+\omega_{pi}^2\right)\right)\delta\psi=0, \label{eq:disp21}
\end{equation}
for linear wave perturbations. An alternative first order wave equation 
follows from (\ref{eq:disp18}) by noting:
\begin{equation}
\omega_+=\bar{\omega}\left[1+\frac{k^2 a_2}{2\bar{\omega}^2}+O(k^4)\right], 
\label{eq:disp22}
\end{equation}
where $\bar{\omega}=(\omega_{pe}^2+\omega_{pi}^2)^{1/2}$. The wave equation 
corresponding to (\ref{eq:disp22}) 
is the linear Schroedinger equation:
\begin{equation}
\left[i\derv{t}+\kappa\frac{\partial^2}{\partial x^2} 
-\bar{\omega}\right]\delta\psi=0,  \label{eq:disp23}
\end{equation}
where
\begin{equation}
\kappa=\frac{c_e^2\omega_{pe}^2+c_i^2\omega_{pi}^2}
{2\bar{\omega}^3}, 
\quad \bar{\omega}=\left(\omega_{pe}^2+\omega_{pi}^2\right)^{1/2}. 
\label{eq:disp24}
\end{equation}

\subsubsection{Special case $c_e=c_i$}

In this case, the solutions for $\omega_\pm^2$ in (\ref{eq:disp10}) 
reduce to:
\begin{equation}
\omega_-=\pm k c_e\quad\hbox{and}\quad \omega_+^2=k^2 c_e^2 +\omega_{pe}^2+\omega_{pi}^2. \label{eq:disp25}
\end{equation}
Thus, the ion-acoustic wave in this limit (i.e. $\omega=\omega_-=\pm k c_e$) 
reduces to a non-dispersive sound wave satisfying either of the equations:
\begin{equation}
\left(\derv{t}+c_e \derv{x}\right) \delta\psi=0\quad\hbox{or}\quad 
\left(\derv{t}-c_e \derv{x}\right) \delta\psi=0 \label{eq:disp26}
\end{equation}
corresponding to the forward and backward sound waves respectively. 

Similarly, the Langmuir wave branch solution reduces to the dispersion 
equation:
\begin{equation}
\omega_+^2=k^2 c_e^2 +\omega_{pe}^2+\omega_{pi}^2. \label{eq:disp27}
\end{equation}
The corresponding wave equation is the Klein Gordon equation:
\begin{equation}
\left(\frac{\partial^2}{\partial t^2}-c_e^2 \frac{\partial^2}{\partial x^2}
+(\omega_{pe}^2+\omega_{pi}^2)\right)\delta\psi=0. \label{eq:disp28}
\end{equation}
Equation (\ref{eq:disp28}) shows that wave dispersion for the ion-acoustic 
wave is intrinsically linked to the difference $c_e^2-c_i^2$ of the squares of 
the electron and ion sound speeds. The dispersionless limit $c_e\to c_i$ 
for the nonlinear version of the linear KdV equation (\ref{eq:disp15}) 
is not straightforward. 
The dispersionless limit of the KdV equation has been investigated by 
Levermore (1988). 

\subsubsection{Short wavelength limit}
At short wavelengths ($k\to\infty$) the dispersion equation 
solutions (\ref{eq:disp10}) reduce to:
\begin{align}
&\omega_+^2=k^2 c_e^2+\omega_{pe}^2
+\frac{\omega_{pe}^2 \omega_{pi}^2}{(c_e^2-c_i^2)} \frac{1}{k^2}
+O\left(\frac{1}{k^4}\right), \label{eq:disp29}\\
&\omega_-^2=k^2 c_i^2+\omega_{pi}^2
-\frac{\omega_{pe}^2 \omega_{pi}^2}{(c_e^2-c_i^2)} \frac{1}{k^2}
+O\left(\frac{1}{k^4}\right). \label{eq:disp30}
\end{align}
It is straightforward to write down partial differential equations (or 
integral equations) corresponding to the dispersion equation expansions 
(\ref{eq:disp29}) and (\ref{eq:disp30}) using the Fourier  correspondence:
$i\partial/\partial t\to \omega$ and $i\partial/\partial x\to -k$. 

Examples of fully nonlinear travelling waves 
of the ion-acoustic or Langmuir wave type are investigated in Section 5. 
In the next section we describe a multi-symplectic Hamiltonian formulation 
of the nonlinear travelling waves for the ion-acoustic and Langmuir 
travelling waves.

\section{Multi-Symplectic Travelling Waves}

The multi-fluid plasma system described by (\ref{eq:ham1})-(\ref{eq:ham5}) 
is a Hamiltonian system (e.g. \cite{Spencer82},
\cite{SpencerKaufman82} 
and \cite{Holm83}). The non-canonical 
Poisson bracket for the 2 fluid electron-proton two fluid plasma described 
by (\ref{eq:ham1})-(\ref{eq:ham5}) can be used to write the evolution 
equations for the system in terms of the non-canonical Eulerian physical 
variables using the Poisson bracket. We use a direct approach to write the 
travelling wave solutions in a canonical Hamiltonian form.

For travelling wave solutions of (\ref{eq:ham1})-(\ref{eq:ham5}), we look for 
solutions of the form $\psi^\alpha=\psi^\alpha(\xi)$  where the physical 
variables $\psi^\alpha$ depend only on the travelling wave variable 
$\xi=x-\lambda t$ where $\lambda$ is the speed of the travelling wave. 
Subsitution of the ansatz $\psi^\alpha=\psi^\alpha(\xi)$ into the governing 
equations (\ref{eq:ham1})-(\ref{eq:ham5}) gives the system of ordinary 
differential equations:
\begin{align}
\frac{d}{d\xi}(n_e u_e)=&\frac{d}{d\xi}(n_p u_p)=0, \label{eq:ham7}\\
u_e\frac{du_e}{d\xi}=&-\frac{1}{m_e n_e} \frac{dp_e}{d\xi}-\frac{e}{m_e} E_x, 
\label{eq:ham8}\\
u_p \frac{du_p}{d\xi}=&-\frac{1}{m_p n_p} \frac{dp_p}{d\xi}
+\frac{e}{m_p} E_x, \label{eq:ham9}\\
\frac{d}{d\xi}J=&\frac{d}{d\xi}[e(n_pu_p-n_eu_e)]=0, \label{eq:ham10}\\
\varepsilon_0\frac{dE_x}{d\xi}=&e(n_p-n_e), \label{eq:ham11}
\end{align}
Here $u_e$ and $u_p$ now refer to the $x$-component of the fluid velocity in 
the travelling waves frame, and $E_x$ is the $x$-component of the electric 
field.  

Integration of the number density continuity equations and the current 
continuity equation (\ref{eq:ham9}) gives the 
integrals:
\begin{equation}
n_pu_p=n_e u_e=j, \label{eq:ham12}
\end{equation}
where $j$ is the common integration constant for the particle fluxes, that 
ensures that the electric current $J$ is zero. Multiplying (\ref{eq:ham8}) 
by $m_e n_e$ and (\ref{eq:ham9}) by $m_pn_p$, adding the two 
resultant equations, and using Poisson's equation to integrate the electric 
field term $eE_x(n_p-n_e)$ gives the $x-$momentum integral for the system as:
\begin{equation}
j(m_e u_e+m_p u_p)+p_e+p_p-\varepsilon_0 \frac{E_x^2}{2}=P_x, 
\label{eq:ham13}
\end{equation}
where $P_x$ is the total $x$-momentum integration constant. Similarly 
multiplying (\ref{eq:ham8}) by $m_e n_e u_e$ and (\ref{eq:ham9}) by 
$m_p n_p u_p$, adding the two equations, and integrating the 
resultant equation with respect to $x$ 
gives the total energy flux conservation equations for the system in the form:
\begin{equation}
 j\left(\frac{1}{2}m_e u_e^2+\frac{\gamma_e p_e}{(\gamma_e-1) n_e}\right)
+j\left(\frac{1}{2}m_p u_p^2+\frac{\gamma_p p_p}{(\gamma_p-1) n_p}\right) 
=\varepsilon, \label{eq:ham14}
\end{equation}
In the integration of (\ref{eq:ham14}) we assumed polytropic equations of 
state for the electron and proton gases, pressures, i.e. 
\begin{equation}
p_e=p_{e0}\left(\frac{n_e}{n_{e0}}\right)^{\gamma_e}
\quad\hbox{and}\quad p_p=p_{p0} \left(\frac{n_p}{n_{p0}}\right)^{\gamma_p}.
\label{eq:ham15}
\end{equation}
Using (\ref{eq:ham11}) and the equations of state (\ref{eq:ham15}) allows 
the electron and proton $x-$momentum equations to be written in the forms:
\begin{equation}
\frac{(u_e^2-c_e^2)}{u_e}\frac{du_e}{d\xi}=-\frac{e}{m_e} E_x, 
\quad\hbox{and}\quad
\frac{(u_p^2-c_p^2)}{u_p}\frac{du_p}{d\xi}=\frac{e}{m_p} E_x, \label{eq:ham16}
\end{equation}
where $c_e^2=\gamma_e p_e/(n_e m_e)$ and $c_p^2=\gamma_p p_p/(n_p m_p)$ 
are the squares of the electron and proton sound speeds respectively. 
Equations (\ref{eq:ham16}) can also be written in the form:
\begin{equation}
\frac{d\varepsilon_e}{d\xi}=-\frac{e}{m_e}E_x,\quad\hbox{and}\quad 
\frac{d\varepsilon_p}{d\xi}=\frac{e}{m_p}E_x, \label{eq:ham17}
\end{equation}
where 
\begin{equation}
\varepsilon_e=\frac{1}{2}u_e^2+\frac{c_e^2}{\gamma_e-1}
\quad\hbox{and}\quad
\varepsilon_p=\frac{1}{2}u_p^2+\frac{c_p^2}{\gamma_p-1}, \label{eq:ham18}
\end{equation}
are the kinetic and enthalpy contributions to the normalized energy fluxes  
of the electron and proton gases. Multiplying the first 
equation (\ref{eq:ham17}) by $jm_e$  and the second equation (\ref{eq:ham17}) 
by $j m_p$ and integrating with respect to $x$ gives the total energy 
equation (\ref{eq:ham14}). 

\subsection{Variational Principles}
\begin{proposition}\label{prop1}
The electron-proton multi-fluid travelling waves 
described by (\ref{eq:ham7})-(\ref{eq:ham18}) are described by the 
stationary point conditions for the variational functional:
\begin{equation}
{\cal A}=\int L d\xi, \label{eq:v1}
\end{equation}
where the Lagrangian density $L$ is:
\begin{equation}
L=E_x \frac{d\tilde{\varepsilon}_e}{d\xi}-[\Pi_x(u_e,u_p,E_x)-P_x]+
\lambda \left(\tilde{\varepsilon}_e+\tilde{\varepsilon}_p
-\tilde{\varepsilon}\right). \label{eq:v2}
\end{equation}
Here
\begin{equation}
\tilde{\varepsilon}_e=\frac{\varepsilon_0 m_e}{e}\varepsilon_e, 
\quad
\tilde{\varepsilon}_p=\frac{\varepsilon_0 m_p}{e}\varepsilon_p, \quad 
\tilde{\varepsilon}=\frac{\varepsilon_0}{j e}\varepsilon, \label{eq:v3}
\end{equation}
are the normalized electron and proton energy fluxes 
defined in (\ref{eq:ham18}), and $\varepsilon$ is the total energy integration 
constant in (\ref{eq:ham14}), $j=n_e u_e=n_pu_p$ and 
$\Pi_x(u_e,u_p,E_x)=P_x$ 
is the $x$-momentum integral (\ref{eq:ham13}). The Lagrange multiplier 
$\lambda$  ensures that the energy conservation integral (\ref{eq:ham14}) 
is satisfied. 
\end{proposition}
\begin{proof}
The stationary point conditions $\delta{\cal A}/\delta\lambda=0$, 
$\delta{\cal A}/\delta u_p=0$, $\delta{\cal A}/\delta u_e=0$ give 
the equations:
\begin{align}
{\cal A}_\lambda=&\tilde{\varepsilon}_e+\tilde{\varepsilon}_p
-\tilde{\varepsilon}=0, \label{eq:v4}\\
{\cal A}_{u_p}=&m_p\left(1-\frac{c_p^2}{u_p^2}\right)
\left(\lambda u_p-j\right)=0, \label{eq:v5}\\
{\cal A}_{u_e}=&m_e\left(1-\frac{c_e^2}{u_e^2}\right)
\left(-\frac{\varepsilon_0}{e}\frac{d E_x}{d\xi}
-(n_e-\lambda)\right)=0, \label{eq:v6}
\end{align}
where we use the notation
 ${\cal A}_{\psi}\equiv\delta {\cal A}/\delta \psi$ for the variational 
derivative of ${\cal A}$ with respect to the variable $\psi$. 
Equation (\ref{eq:v4}) is equivalent to the energy 
integral (\ref{eq:ham14}). Equation (\ref{eq:v5}) determines 
the Lagrange multiplier $\lambda$, i.e., 
\begin{equation}
\lambda=\frac{j}{u_p}=n_p. \label{eq:v7}
\end{equation}
Using $\lambda=n_p$  in (\ref{eq:v6}) gives the 
Poisson equation (\ref{eq:ham11}), i.e., 
\begin{equation}
\varepsilon_0\frac{dE_x}{d\xi}=e(n_p-n_e). \label{eq:v8}
\end{equation}
The variational equation 
\begin{equation}
{\cal A}_{E_x}=\frac{\varepsilon_0 m_e}{e}\left(\frac{d\varepsilon_e}{d\xi}
+\frac{e}{m_e} E_x\right)=0, \label{eq:v9}
\end{equation}
gives the electron momentum equation in (\ref{eq:ham17}). 
The remaining equations of the system follow from the above variational 
equations. For example, differentiation of the energy integral (\ref{eq:v4}) 
with respect to $\xi$, coupled with the electron momentum 
equation (\ref{eq:v9}) gives the proton momentum equation in (\ref{eq:ham17}). 
The total momentum equation (\ref{eq:ham13}) follows by the integration 
of a suitable combination of the electron and proton momentum equations  
as indicated in (\ref{eq:ham18}) et seq.. This completes the proof.
\end{proof}

\begin{proposition}\label{prop2}
The electron energy equation in (\ref{eq:ham17}) and Poisson's equation 
(\ref{eq:ham11}) can be written in the Hamiltonian form:
\begin{align}
&\frac{d\tilde{\varepsilon}_e}{d\xi}=\deriv{P_x}{E_x}, \label{eq:ham19}\\
&\frac{dE_x}{d\xi}=-\deriv{P_x}{\tilde{\varepsilon}_e}, \label{eq:ham20}
\end{align}
where the partial derivatives of the total momentum function $P_x$ are taken 
keeping the total energy flux $\varepsilon$ in (\ref{eq:ham14}) constant, 
and 
\begin{align}
&\tilde{\varepsilon}_e=\frac{\varepsilon_0 m_e}{e}\varepsilon_e\equiv 
\frac{\varepsilon_0 m_e}{e}\left[\frac{1}{2}u_e^2
+\frac{c_e^2}{\gamma_e-1}\right],\nonumber\\ 
&\tilde{\varepsilon}_p=\frac{\varepsilon_0 m_p}{e}\varepsilon_p\equiv 
\frac{\varepsilon_0 m_p}{e}\left[\frac{1}{2}u_p^2
+\frac{c_p^2}{\gamma_p-1}\right], \label{eq:ham21}
\end{align}
are re-normalized energy fluxes of the electron and proton gases.
\end{proposition}
\begin{proof}
To prove (\ref{eq:ham19}) note from (\ref{eq:ham17}) and (\ref{eq:ham13}) that
\beqn
\frac{d\tilde{\varepsilon}_e}{d\xi}=-\varepsilon_0 E_x\quad \hbox{and}\quad 
\deriv{P_x}{E_x}=-\varepsilon_0 E_x. \label{eq:ham22}
\eeqn
These two equations imply (\ref{eq:ham19}). 

To prove (\ref{eq:ham20}) note that as the total energy flux $\varepsilon$
is held constant, that $u_p=u_p(u_e)$ is a function of $u_e$. Hence 
\begin{equation}
\deriv{P_x}{\tilde{\varepsilon}_e}
=\deriv{P_x}{u_e}/\frac{d\tilde{\varepsilon}_e}{du_e}, \quad 
\deriv{P_x}{u_e}=\left(\deriv{P_x}{u_e}\right)_{u_p}
+\deriv{P_x}{u_p}\frac{du_p}{du_e}. \label{eq:ham23}
\end{equation}
Separately evaluating the different terms in (\ref{eq:ham23}) we obtain:
\begin{align}
&\frac{d\tilde{\varepsilon}_e}{d u_e}=\frac{\varepsilon_0 m_e}{e} 
\frac{d\varepsilon_e}{du_e}
=\frac{\varepsilon_0 m_e}{e}\frac{(u_e^2-c_e^2)}{u_e}, \nonumber\\
&\left(\deriv{P_x}{u_e}\right)_{u_p}
=\frac{jm_e(u_e^2-c_e^2)}{u_e^2}, 
\quad \deriv{P_x}{u_p}=\frac{jm_p(u_p^2-c_p^2)}{u_p^2}, \nonumber\\
&\frac{du_p}{du_e}=-\frac{m_eu_p(u_e^2-c_e^2)}{m_p u_e(u_p^2-c_p^2)}, \quad
\deriv{P_x}{u_e}=\frac{j m_e(u_e^2-c_e^2)}{u_e^2}
\left(1-\frac{n_p}{n_e}\right). \label{eq:ham24}
\end{align}
Using the results (\ref{eq:ham24}) in (\ref{eq:ham23}) gives 
\begin{equation}
\deriv{P_x}{\tilde{\varepsilon}_e}=\frac{e(n_e-n_p)}{\varepsilon_0}
=-\frac{dE_x}{d\xi}. \label{eq:ham25}
\end{equation}
This completes the proof. 
\end{proof}

\begin{proposition}\label{prop3}
The electron proton multi-fluid travelling waves 
(\ref{eq:ham7})-(\ref{eq:ham18}) can be obtained from the stationary point
conditions for the variational functional:
\begin{equation}
{\cal A}=\int L_2\ d\tau, \label{eq:vd1}
\end{equation}
where the Lagrangian density $L_2$ has the form:
\begin{equation}
L_2=E_x\frac{d\tilde{\varepsilon}_e}{d\tau}+\left[E(u_e,u_p)
-\varepsilon\right] +\lambda\left[\Pi_x(u_e,u_p,E_x)-P_x\right]. 
\label{eq:vd2}
\end{equation}
In (\ref{eq:vd2}) 
\begin{equation}
E(u_e,u_p)=\varepsilon, \quad \Pi_x(u_e,u_p,E_x)=P_x, \label{eq:vd3}
\end{equation}
are the total energy and momentum conservation laws (\ref{eq:ham14}) and 
(\ref{eq:ham13}) respectively, and 
$\tilde{\varepsilon}_a=\varepsilon_0 (m_a\varepsilon_a)/e$ ($a=e,p$). The 
Lagrange multiplier $\lambda$ in (\ref{eq:vd2}) ensures that the momentum 
conservation equation (\ref{eq:ham13}) is satisfied. The parameter
\begin{equation}
\tau=\int_{\xi_0}^{\xi} \frac{d\xi}{u_p}, \label{eq:vd4}
\end{equation}
is the time travelled by the proton fluid from some fiducial point $\xi_0$ 
to the position $\xi$ in the wave frame. Thus, 
\begin{equation}
\frac{d\tilde{\varepsilon}_e}{d\tau}=u_p \frac{d\tilde{\varepsilon}_e}{d\xi}. 
\label{eq:vd5}
\end{equation}
\end{proposition}
\begin{proof}
Evaluating $\delta{\cal A}/\delta u_p$, $\delta{\cal A}/\delta E_x$, 
and $\delta {\cal A}/\delta u_e$ gives the equations:
\begin{align}
&{\cal A}_{u_p}=j m_p\left(1-\frac{c_p^2}{u_p^2}\right)
\left(\lambda+u_p\right)=0, \label{eq:vd6}\\
&{\cal A}_{E_x}=\frac{d\tilde{\varepsilon}_e}{d\tau}
-\lambda(\varepsilon_0 E_x)=0,  \label{eq:vd7}\\
&{\cal A}_{u_e}=-\left(\frac{\varepsilon_0 m_e}{e}\right) 
u_e\left(1-\frac{c_e^2}{u_e^2}\right)
\left[\frac{dE_x}{d\tau}-u_p\frac{e}{\varepsilon_0} (n_p-n_e)\right]=0. 
\label{eq:vd8}
\end{align}
Solving (\ref{eq:vd6}) for $\lambda$ gives:
\begin{equation}
\lambda=-u_p. \label{eq:vd9}
\end{equation}
Equations (\ref{eq:vd7}) and (\ref{eq:vd8}) reduce to:
\begin{align}
&\frac{d\tilde{\varepsilon}_e}{d\tau}=-u_p(\varepsilon_0 E_x)
\quad\hbox{or}\quad 
\frac{d\tilde{\varepsilon}_e}{d\xi}=-(\varepsilon_0 E_x), \label{eq:vd10}\\
&\frac{dE_x}{d\tau}=u_p\frac{e}{\varepsilon_0} (n_p-n_e)\quad\hbox{or}\quad 
\frac{dE_x}{d\xi}=\frac{e}{\varepsilon_0} (n_p-n_e). \label{eq:vd11}
\end{align}
Thus, (\ref{eq:vd7})-(\ref{eq:vd8}) or (\ref{eq:vd10})-(\ref{eq:vd11}) 
are equivalent to the electron momentum equation in (\ref{eq:ham17}) 
and to Poisson's equation (\ref{eq:ham11}). 
The equation $\delta{\cal A}/\delta\lambda=0$ gives the total $x$-momentum 
conservation law (\ref{eq:ham13}). The above equations and the total momentum 
conservation equation (\ref{eq:ham13}) imply the two-fluid equations 
(\ref{eq:ham7})-(\ref{eq:ham17}) for the travelling waves. This completes the
proof.
\end{proof}
\begin{remark}
The present formulation of the action for the travelling waves 
clearly differs from that in Proposition \ref{prop1}, where 
the evolution operator is the spatial operator $d/d\xi$ which 
is different from the time evolution operator $d/d\tau=u_p d/d\xi$ 
in the present formulation.  
\end{remark}

\begin{proposition}
The electron momentum equation  (\ref{eq:ham17}) and Poisson's equation 
(\ref{eq:ham11}) can be written in the Hamiltonian form:
\begin{align}
&\frac{d\tilde{\varepsilon}_e}{d\tau}=-\deriv{E}{E_x}, \label{eq:vd12}\\
&\frac{dE_x}{d\tau}=\deriv{E}{\tilde{\varepsilon}_e}, \label{eq:vd13}
\end{align}
where
\begin{equation}
\tau= \int^\xi \frac{d\xi}{u_p}, \label{eq:vd14}
\end{equation}
is the evolution variable and $E(u_p,u_e)=\varepsilon$ is the energy integral 
(\ref{eq:ham14}). The partial derivatives of $E(u_e,u_p)$ are taken keeping the momentum integral $\Pi_x(u_e,u_p,E_x)$ constant 
where $\Pi_x(u_e,u_p,E_x)=P_x$ is the momentum integral (\ref{eq:ham13}). 
\end{proposition}

\begin{proof}
Since $\Pi_x$ is constant, then the momentum 
integral $\Pi_x(u_e,u_p,E_x)=P_x$ can be solved for $u_p=u_p(u_e,E_x)$ 
as a function of $u_e$ and $E_x$. To compute $\partial E/\partial E_x$ in 
(\ref{eq:vd12}) we note
\begin{equation}
\left(\deriv{E}{E_x}\right)_{u_e}
=\deriv{E}{u_p} \deriv{u_p}{E_x}, \label{eq:vd15}   
\end{equation}
Using the results:
\begin{equation}
\deriv{u_p}{E_x}=\frac{\varepsilon_0 E_x u_p}{n_pm_p(u_p^2-c_p^2)}, 
\quad \left(\deriv{E}{u_p}\right)_{u_e}
=n_p m_p (u_p^2-c_p^2), \label{eq:vd16}
\end{equation}
(\ref{eq:vd15}) reduces to:
\begin{equation}
\deriv{E}{E_x}=u_p\varepsilon_0 E_x. \label{eq:vd17}
\end{equation}
The electron momentum equation (\ref{eq:ham17}) is:
\begin{equation}
\frac{d\varepsilon_e}{d\xi}=-\frac{e}{m_e} E_x. \label{eq:vd18}
\end{equation}
Equations (\ref{eq:vd17}) and (\ref{eq:vd18}) imply Hamilton's  equation 
(\ref{eq:vd12}). 

To compute $\partial E/\partial\tilde{\varepsilon}_e$ we note:
\begin{equation}
\deriv{E}{\tilde{\varepsilon}_e}
=\deriv{E}{u_e}/\frac{d\tilde{\varepsilon}_e}{du_e},
\quad \left(\deriv{E}{u_e}\right)_{E_x}
=\left(\deriv{E}{u_e}\right)_{u_p}+
\deriv{E}{u_p}\left(\deriv{u_p}{u_e}\right)_{E_x}. \label{eq:vd19}
\end{equation}
Separately computing the terms in (\ref{eq:vd19}) we obtain:
\begin{align}
&\frac{d\tilde{\varepsilon}_e}{d u_e}=\frac{\varepsilon_0 m_e}{e} 
\frac{d\varepsilon_e}{du_e}
=\frac{\varepsilon_0 m_e}{e}\frac{(u_e^2-c_e^2)}{u_e}, \quad
\left(\deriv{E}{u_e}\right)_{u_p}
=\frac{jm_eu_e(u_e^2-c_e^2)}{u_e^2}, \nonumber\\
&\deriv{E}{u_p}=\frac{jm_pu_p(u_p^2-c_p^2)}{u_p^2}, 
\quad \frac{du_p}{du_e}=-\frac{m_eu_p(u_e^2-c_e^2)}{m_p u_e(u_p^2-c_p^2)},
\nonumber\\
&\deriv{E}{u_e}=\frac{j m_e (u_e^2-c_e^2)}{u_e}
\left(1-\frac{n_e}{n_p}\right). \label{eq:vd20}
\end{align}
Using (\ref{eq:vd20}) in (\ref{eq:vd19}) gives:
\begin{equation}
\deriv{E}{\tilde{\varepsilon}_e}=u_p 
\left[\frac{e}{\varepsilon_0}(n_p-n_e)\right] 
\equiv u_p \frac{dE_x}{d\xi}=\frac{dE_x}{d\tau}, \label{eq:vd21}
\end{equation}
which establishes the Hamiltonian equation (\ref{eq:vd13}) for 
$dE_x/d\tau$. This completes the proof.
\end{proof}

\section{Integral form of steady-state solitary wave signatures in a multi-fluid plasma}
We consider a steady-state, multi-fluid plasma system, viewed in the wave 
frame, where the bulk velocity of each plasma species 
is $\hat{x}u_0$ as $x\to -\infty$ and the wave form appears stationary 
so that $\p/\p t \to 0$. 
Thus the momentum equations and Poisson's equation are
\begin{align}
\frac{d}{dx}\left(p_j + m_jn_ju_j^2\right) &= q_jn_jE_x;  \label{eqn:SSmom} \\
\frac{dE_x}{dx} &= 4\pi \sum_{j}q_jn_j. \label{eqn:SSPoisson} 
\end{align}
By adopting the McKenzie approach of writing dimensionless velocities as $u_j\to u_j/u_0$ and using the continuity $$n_ju_j = n_{j0}u_0$$ and adiabatic energy $$p_j u_j^{\gamma_j} = p_{j0}u_{0}^{\gamma_j}$$ relations, the momentum equation (\ref{eqn:SSmom}) can be written as
\be\label{eqn:dimumom}
\frac{d}{dx}\left[ \frac{u_j^2}{2}+\frac{u_j^{1-\gamma_j}}{(\gamma_j-1)M_j^2} \right]=\frac{q_jE_x}{m_ju_0^2},
\ee
where the 0-subscript denotes a constant upstream state and
\[
M_j = \frac{u_0}{c_{j0}}, \quad c_{j0}^2 = \frac{\gamma_j p_{j0}}{m_j n_{j0}},
\] is the $j$'th species sound speed Mach number. Integrating equation (\ref{eqn:SSmom}), using (\ref{eqn:SSPoisson}), yields
\begin{align}
E_x &= \pm\sqrt{8\pi u_0^2\sum_{j}m_jn_{j0}P(u_j)}, \label{eqn:ExsqofP} \\
P(u_j) &= (u_j-1)+\frac{1}{\gamma_j M_j^2}(u^{-\gamma_j}-1), \label{eqn:Pofu}
\end{align}
which gives the electric field in terms of conservation of total momentum and where the generalized momentum functions $P_j$ are composed of two terms: the first associated with the dynamic pressure and the second with the thermal pressure of the $j$'th species. Introducing the electrostatic potential $d\phi/dx=-E_x$ and again integrating the momentum equation (\ref{eqn:dimumom}) yields conservation of energy 
\begin{align}
\frac{m_ju_0^2}{q_j}\varepsilon(u_j) &= \frac{m_iu_0^2}{q_i}\varepsilon(u_i) = -\phi, \label{eqn:EgPropEg} \\
\varepsilon(u_j) &= \frac{u_j^2-1}{2} + \frac{(u_j^{1-\gamma_j}-1)}{(\gamma_j-1)M_j^2}, \label{eqn:Egofu}
\end{align}
where the first equation (\ref{eqn:EgPropEg}) expresses the energy proportionality 
between the $j$'th and $i$'th species, and (\ref{eqn:Egofu}) defines the Bernoulli 
energy function $\varepsilon_j$, composed of two terms: the first associated with dynamic energy 
and the second with the enthalpy of the species. The relations (\ref{eqn:ExsqofP}) 
through (\ref{eqn:Egofu}) yield the wave amplitude and necessary conditions for a solitary 
wave to exist.  As summarized in Table \ref{tab:acdc}, the energy proportionality 
relations (\ref{eqn:EgPropEg}) and (\ref{eqn:Egofu}), can be employed to determine whether a 
fluid species will accelerate or decelerate from the initial point
 (\cite{McKenzie02}, \cite{Verheest04b}) 
\begin{table}
\begin{center}
\begin{tabular}{l|c|c|c} 
\hline\hline
charge & Mach number & potential & velocity from IP \\
\hline
\multirow{4}{*}{ions ($q>0$)} & $M_i>1$ & $\phi >0$ & decelerates \\
                              & $M_i>1$ & $\phi <0$ & accelerates \\
                              & $M_i<1$ & $\phi >0$ & accelerates \\
                              & $M_i<1$ & $\phi <0$ & decelerates \\
\hline
\multirow{4}{*}{electrons ($q<0$)} & $M_e>1$ & $\phi >0$ & accelerates \\
                                   & $M_e>1$ & $\phi <0$ & decelerates \\
                                   & $M_e<1$ & $\phi >0$ & decelerates \\
                                   & $M_e<1$ & $\phi <0$ & accelerates \\
\hline\hline
\end{tabular}
\end{center}
\caption{Velocity behavior  at the initial point (IP), acceleration or deceleration as 
determined from equation (\ref{eqn:EgPropEg}), for ions and electron species. 
$\phi>0$ indicates a potential hill whereas $\phi<0$ indicates 
a potential well. 
}\label{tab:acdc}
\end{table}

The functions $P_j$ and $\varepsilon_j$ can be used to calculate solitary wave signatures. Using equations (\ref{eqn:EgPropEg}) and (\ref{eqn:Egofu}), the velocity of any plasma species $i$ can be expressed as a function of the $j$'th species in the form $u_i=u_i(u_j)$. Thus all the $i$'th momentum functions can be written in terms of the $j$'th species velocity, $P_i=P(u_i(u_j))$, so that the electric field can also be expressed as a function of the same single species velocity: $E_x=E_x(u_j)$. The above procedure can always (in principle) be carried out, although it is sometimes more convenient to do so numerically. 

We write the structure equation (\ref{eqn:dimumom}) as 
\be\label{eqn:intsol}
x = \int_{u_{c}}^{u_j}\frac{u \left( 1-\frac{1}{M_j^2u^{\gamma_j+1}} \right)\mathrm{d}u}{q_jE_x/m_ju_0^2},
\ee
where $u_{c}$ is the velocity of species $j$ at the center of the wave where $x\equiv 0$. Note that since $E_x=E_x(u_j)$, equation (\ref{eqn:intsol}) can be used directly to integrate solitary wave signatures for a plasma system composed of any number of different fluid species. Thus equations (\ref{eqn:ExsqofP}) through (\ref{eqn:Egofu}) along with the structure equation (\ref{eqn:intsol}) form a complete description for solitary waves in an adiabatic, multi-fluid plasma. However the determination of existence conditions and integration of the structure equation is not generally straightforward and should be treated carefully on a case-by-case basis.

\subsection{Integration of a plasma composed of cold protons and hot electrons}
In this sub-section we illustrate the utility of the closed system of multi-fluid conservation laws by considering a two-fluid plasma composed of cold (highly supersonic, $T_{p}=0$) protons and hot (subsonic) electrons. The utility of the structure equation (\ref{eqn:intsol}) is exploited to integrate exact solutions, revealing properties particularly important for particle reflection--the width and amplitude of the wave, both of which are found to depend critically on the Mach number of the incident flow.  

On using the normalizations,
\begin{align}\label{eqn:SSdimvars}
N &= n_p/n_0, & n &= n_e/n_0, & U &= u_p/u_0, & u &= u_e/u_0, \nonumber \\
E_x &\to E_x/E_0, & x & \to x/\lambda_D, & E_0 &= \frac{k T_{e0}}{e\lambda_D}, & \lambda_D^2 &= \frac{kT_{e0}}{4\pi n_0e^2}, \nonumber \\
\alpha_0 &= m_e/m_p, & \alpha_1^2 &= \frac{\alpha_0}{\gamma M_{e}^2}, 
\end{align}
equation (\ref{eqn:intsol}) becomes
\begin{align}\label{eqn:dimintsol}
x &= \int_{u_{c}}^{u}f_{\pm}(u)du; \\ 
f_{\pm}(u) &= \frac{\left(\frac{1}{M_e^2 u^{\gamma+1}}-1\right)\alpha_0 u}{\pm\sqrt{2\alpha_1^2\left(U-1+\alpha_0(u-1)+\alpha_1^2(u^{-\gamma}-1) \right)}}, \nonumber 
\end{align}
where, in view of equation (\ref{eqn:EgPropEg}) and (\ref{eqn:Egofu}),
\be
U=\sqrt{1+\alpha_0(1-u^2)+\frac{2\gamma\alpha_1^2}{\gamma-1}(1-u^{1-\gamma})}
\ee 
and $u_{c}$ is the electron velocity at the center of the solitary wave where $x=0$. This expression explicitly retains the electron to proton mass ratio $\alpha_0$.

The ion acoustic Mach number $M_{ia}$ is defined  by the equation:
\be
M_{ia}^2=\frac{u_0^2}{c_{ia}^2}\quad
\hbox{where}\quad c_{ia}^2=\frac{c_i^2 m_i+c_e^2 m_e}{m_i+m_e}. 
\label{eq:5.12}
\ee
Expression (\ref{eq:5.12}) for $c_{ia}^2$ is equivalent to the result 
(\ref{eq:disp13}) for $c_{ia}^2$. It is also useful to define the Mach number 
$M_{ep}^2$ used by \cite{McKenzie02}:
\be
M_{ep}^2=\frac{u_0^2}{c_{ep}^2}\quad\hbox{where}
\quad c_{ep}^2=\frac{\gamma_e p_e}{n m_p}\equiv 
\frac{\gamma_e k_B T_e}{m_p}, \label{eq:5.13} 
\ee
The relationship between $M_{ia}^2$ and $M_{ep}^2$ is:
\be
M_{ia}^2=(1+\alpha_0) M_{ep}^2. \label{eq:5.14}
\ee
Note if $\alpha_0=m_e/m_p<<1$ then $M_{ia}^2\approx M_{ep}^2$. 

The Mach number regime over which solitary wave solutions exist can be 
determined by the methods used by \cite{McKenzie02}) and 
\cite{Verheest04b}.
The ion-acoustic Mach number $M_{ia}$ is restricted to the range:
\be
1<M_{ia}^2<M_{max}^2. \label{eq:5.15}
\ee

The lower limit $M_{ia}=1$ in (\ref{eq:5.15}) corresponds to a weakly nonlinear ion acoustic wave, i.e. to a linear ion-acoustic wave. For a weakly nonlinear, 
long wavelength ion acoustic wave to have spatial wave growth and decay for 
the wave envelope like $\exp(\pm\kappa x)$, requires $M_{ia}^2>1$ (see below). 
Note that in a dispersive soliton or a solitary wave the nonlinearity 
is balanced by the dispersion. To derive the condition 
$M_{ia}>1$ we consider solutions of the linear dispersion equation (\ref{eq:disp3}) in conjunction with the travelling wave condition $\omega=k u_0$ where 
$u_0$ is the speed of the weak travelling wave (cf. \cite{McKenzie-etal04}). 
  Thus, we look for solutions of the equations:
\be
k^2=\frac{\omega_{pi}^2}{\lambda^2-c_i^2}
+\frac{\omega_{pe}^2}{\lambda^2-c_e^2}\quad\hbox{where}
\quad \lambda=\frac{\omega}{k}=u_0. \label{eq:5.16}
\ee
For the sake of simplicity we assume $c_e^2>c_i^2$ (i.e. hot electrons 
and cooler ions). It is straightforward to prove in this case that 
\be
c_i^2<c_{ia}^2<c_e^2, \label{eq:5.17}
\ee
where the ion acoustic speed $c_{ia}$ is given by (\ref{eq:5.12}). 
Eliminating reference to $\omega$ in (\ref{eq:5.16}) we obtain:
\be
k^2=\frac{(\omega_{pi}^2+\omega_{pe}^2)(u_0^2-c_{ia}^2)}
{(u_0^2-c_i^2)(u_0^2-c_e^2)}.  \label{eq:5.18}
\ee
A sketch of $k^2$ versus $u_0^2$ reveals that if $c_{ia}^2<u_0^2<c_e^2$ 
then $k^2<0$, (i.e. the dispersion equation has pure imaginary solutions for 
$k=i\kappa$). Thus $k^2<0$ if: 
\be
1<M_{ia}^2<\frac{c_e^2}{c_{ia}^2}\equiv 
\frac{c_e^2(1+\alpha_0)}{(c_i^2+c_e^2\alpha_0)}. \label{eq:5.19}
\ee
Hence $M_{ia}^2>1$ is required for spatial wave growth 
and decay as occurs for a soliton envelope. The upper limit on the 
Mach number $M_{ia}^2$ requires  
the use of the fully nonlinear equations of the system. The determination 
of $M_{max}$ is discussed below.

The upper bound, 
$M_{max}$,  in (\ref{eq:5.15})  comes from requiring $E_x=0$ at an 
equilibrium point of the flow (assumed to occur at the center of the wave). 
For a compressive ion acoustic solitary wave, 
the equilibrium point $u_p=u_{eq}$ 
is such that $u_{eq}>u_{sonic}$ where $u_{sonic}=u_0 M_{p}^{-2/(\gamma_p+1)}$
is the value of $u_p$ at the proton sonic point. Thus, the basic equations 
to determine $(M_{ep})_{max}=M_c$ are:
\begin{align}
&E_x^2=\frac{2 u_0^2 n_0}{\varepsilon_0} \left[m_e P_e(\bar{u}_e)
+m_p P_p({\bar u}_p)\right]=0, \nonumber\\
&m_e\varepsilon_e(\bar{u}_e)+m_p \varepsilon_p({\bar u}_p)=0.  \nonumber\\
&u_p=u_{sonic}=u_0 M_{p}^{-2/(\gamma_p+1)}, \label{eq:5.20}
\end{align}
where ${\bar u}_j=u_j/u_0$ ($j=e,p$), and 
\begin{align}
&P_j({\bar u}_j)={\bar u}_j-1+\frac{1}{\gamma_j M_j^2} 
\left({\bar u}_j^{-\gamma_j}-1\right), \quad (j=e,p), \nonumber\\
&\varepsilon_j({\bar u}_j)=\frac{1}{2}\left({\bar u}_j^2-1\right) 
+\frac{1}{(\gamma_j-1) M_j^2} \left({\bar u}_j^{-(\gamma_j+1)}-1\right), 
\quad (j=e,p), \label{eq:5.21}
\end{align} 
are the momentum and energy functions for species $j$ ($j=e,p$). 

For arbitrary $\gamma_e$, $M_{max}$ can be calculated implicitly as done for 
the curves corresponding to $\gamma_e = 5/3$ and $\gamma_e = 1$ of 
Figure \ref{fig:Mvsa0}.

For the case of a cold proton gas ($M_p\to\infty$ or $p_p=0$)
and $\gamma_e=2$, and a finite electron mass ($\alpha\neq 0$),  
(\ref{eq:5.20})-(\ref{eq:5.21}) admit analytical solutions 
for ${\bar u}_e$ and $M_{ep}$ (McKenzie (2002) also obtained 
the analytical solution for $\gamma_e=2$ and $\alpha_0=0$). 

 In the above case, (\ref{eq:5.20}) reduce to:
\begin{align}
&{\bar u}_p-1+\frac{1}{\gamma_e M_{ep}^2}
\left({\bar u}_e^{-\gamma_e}-1\right)
+\alpha_0\left({\bar u}_e-1\right)=0, \label{eq:5.22}\\
&\alpha_0\frac{1}{2}\left({\bar u}_e^2-1\right) 
+\frac{1}{(\gamma_e-1) M_{ep}^2}\left({\bar u}_e^{1-\gamma_e}-1\right)
+\frac{1}{2}\left({\bar u}_p^2-1\right)=0,  \label{eq:5.23}\\
&{\bar u}_p=0\quad \hbox{(proton fluid  sonic point condition)}. 
\label{eq:5.24}
\end{align}

For $\gamma_e=2$, (\ref{eq:5.22})-(\ref{eq:5.23}) has the root
\be
 {\bar u}_e=\frac{3}{2\tilde{M}^2+1} 
\quad\hbox{where}\quad \tilde{M}^2=(1+\alpha_0) M_{ep}^2, \label{eq:5.25}
\ee
and $\tilde{M}^2$ satisfies the equation:
\be
\left(2\tilde{M}^2+1\right)^3-9\left(2\tilde{M}^2+1\right)^2
+\frac{54\alpha_0}{(\alpha_0+1)} \tilde{M}^2=0. \label{eq:5.26}
\ee
Equation (\ref{eq:5.26}) can be written in the form:
\be
\left(\tilde{M}^2-1\right)^3-b\left(\tilde{M}^2-1\right)-b=0\quad\hbox{where}\quad b=\frac{27}{4(1+\alpha_0)}. \label{eq:5.27}
\ee
Using standard formulae for the solution of a cubic equation (
{Abramowitz65}, formula 3.8.2, p. 17), (\ref{eq:5.27}) has one real root 
for $\theta=\tilde{M}^2-1$. This real root gives $\tilde{M}^2\equiv M_{max}^2$:
\be
M_{max}^2=1+s_1+s_2 \quad\hbox{where}\quad 
s_{1,2}=\frac{3}{2}\left[\frac{1}{(1+\alpha_0)} 
\left(1\pm\left(\frac{\alpha_0}{(1+\alpha_0)}\right)^{1/2}
\right)\right]^{1/3}. \label{eq:5.28}
\ee
For the massless electron limit $\alpha_0\to 0$ (\ref{eq:5.28}) gives 
$M_{max}=2$ which agrees with the result of \cite{McKenzie02}.

\begin{figure}
\begin{center}
\includegraphics[width=1.0\textwidth, height=0.45\textheight]{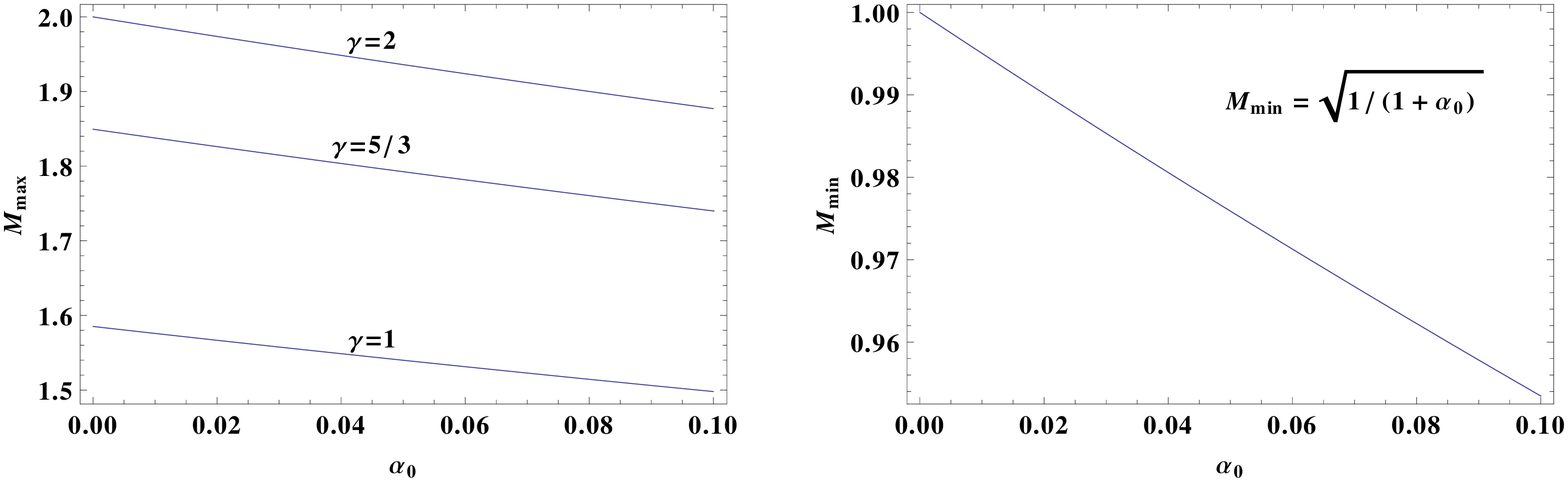}
\end{center}
\caption{Electron-ion Mach number $M_{ep}$ (see (\ref{eq:5.13})) 
 range as a function of electron 
to ion mass ratio, for which solitary wave solutions exist, for 
several choices of the electron adiabatic gas index $\gamma$.}\label{fig:Mvsa0}
\end{figure}

Figure \ref{fig:sgntrM1p4} illustrates a typical solitary wave signature resulting from the integration of equation (\ref{eqn:dimintsol}), the ($f_+$) kernel of which is plotted in Figure \ref{fig:fofu}.
\begin{figure}
\includegraphics[width=1.0\textwidth, height=0.5\textheight]{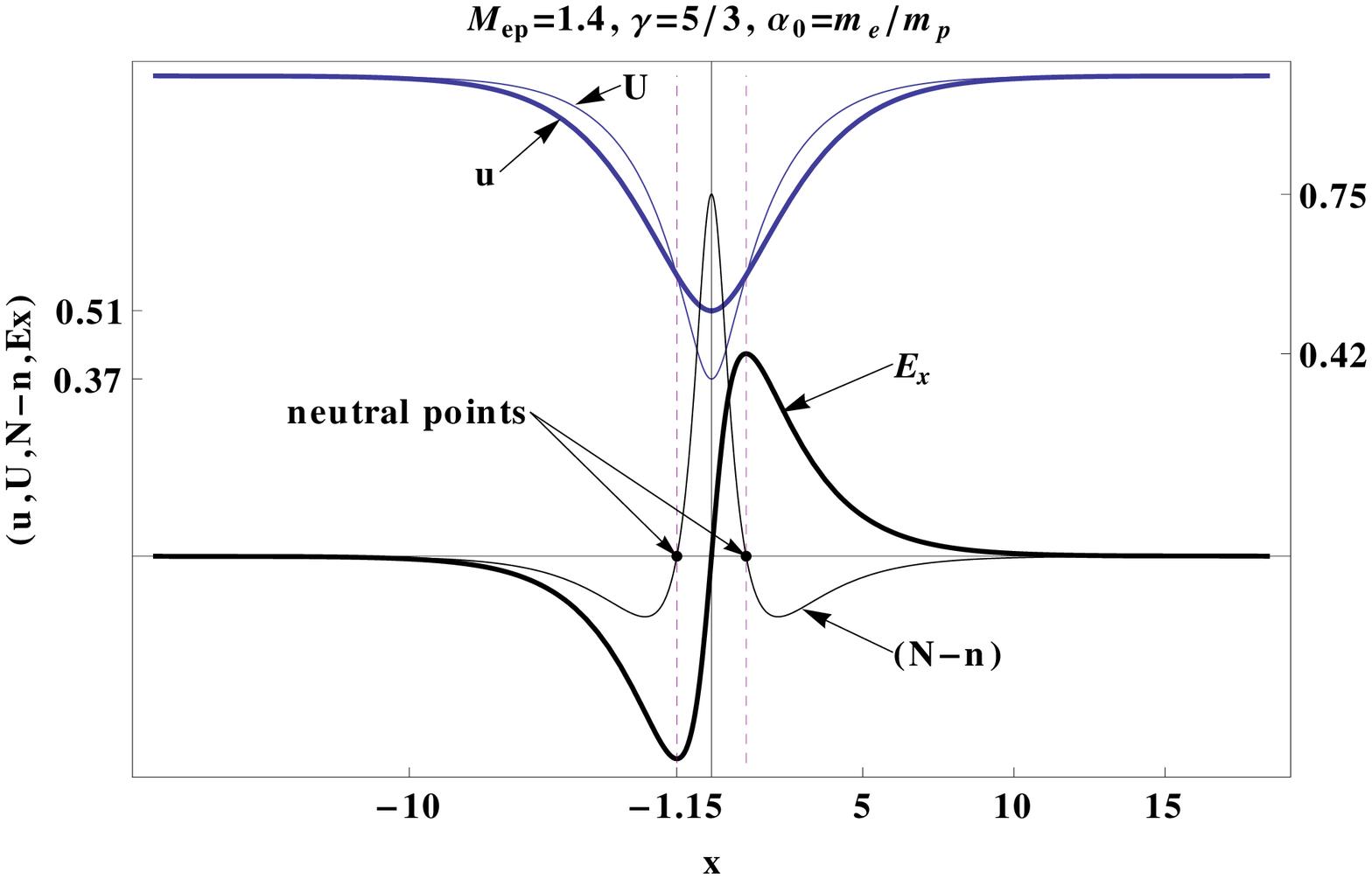}
\caption{Typical ion-acoustic solitary wave signature.}\label{fig:sgntrM1p4}
\end{figure}
Figure \ref{fig:sgntrM1p4} reveals several properties of ion-acoustic solitons. Both fluids are decelerated in the first part of the wave, consistent with a potential hill (see Table \ref{tab:acdc}), with the protons running ahead of the electrons, up to the neutral point where the electrons overtake the protons and the electric field reaches an extrema. At the center of the wave, where velocities are minimum, $E_x=0$ and the charge density reaches a maximum. As suggested by Figure \ref{fig:SolEnvlp}, for zero temperature protons the charge density and slope of the electric field at the wave center actually tend to infinity as $M_{ep}\to M_{max}$. The wave behavior described by Figure \ref{fig:sgntrM1p4} is related to the formation of a cross-shock potential in that the protons, which due to their mass, are able to penetrate deeper into the wave then the electrons, creating charge imbalance and subsequent electro-static forces which act to restore charge neutrality and prevent either species from `running away' in the flow. The charge imbalance (which here tends to infinity at the critical Mach number owing to the zero proton pressure) that occurs throughout the wave is worth emphasizing in that it underlines the care required when imposing conditions of quasi-charge neutrality in plasma-fluid models.

It can be shown (\cite{Verheest04b}) that, solitary waves, associated with potential wells, which result from integrating the $f_{\pm}$ upper-branch (see Figure \ref{fig:fofu}) are unphysical. To see this, note that the area under the curve is positive to the left and negative to the right of the electron sonic point, which means that the solution curve becomes double valued there. The most physically plausible way to find a single valued, upper-branch solution $u=u(x)$, then, is to place $u_c$ at the sonic point and patch together solution branches that satisfy the original ODEs and boundary conditions. However, as can be shown by evaluating $dE_x^2/du_i=0$, $E_x$ has extrema at sonic points and neutral points which contradicts the requirement that $E_x=0$ at the wave center. 
 
Several solitary wave profiles are plotted in Figure \ref{fig:sltrywvs} for the case $\gamma = 5/3$ and $\alpha_0 = 0.0005446$ and for Mach numbers ranging from 1.01 to 1.8. This figure illustrates graphically the nonlinear wave steepening and amplification that occurs with increasing Mach number. Note that near the critical Mach number (here $M_{ep} = 1.84886$) the flow becomes completely `choked' as the gradients of the plasma parameters approach infinity. The charge density curve ($N-n$) for the near critical Mach number ($M_{ep}=1.8$) case has been purposely omitted from the plots since its amplitude (of about 30 in normalized units) is so large that its inclusion would obscure the fine structure of the other charge density curves.

An additional property of solitary waves (perhaps not clearly visible in Figure \ref{fig:sltrywvs}) is that they become more narrow with increasing Mach number. To see this, the full width at half minimum (FWHM) of the electron velocity is plotted verses Mach number in Figure \ref{fig:fwhm}. Evidently, the wave amplitude increases (Figure \ref{fig:SolEnvlp}) and its structure narrows (Figure \ref{fig:fwhm}), both with increasing Mach number. The, not surprising, implication of this behavior is that waves associated with faster flow should be better able to act as particle reflectors.

\begin{figure}
\begin{center}
\includegraphics[width=0.75\textwidth, height=0.75\textheight]{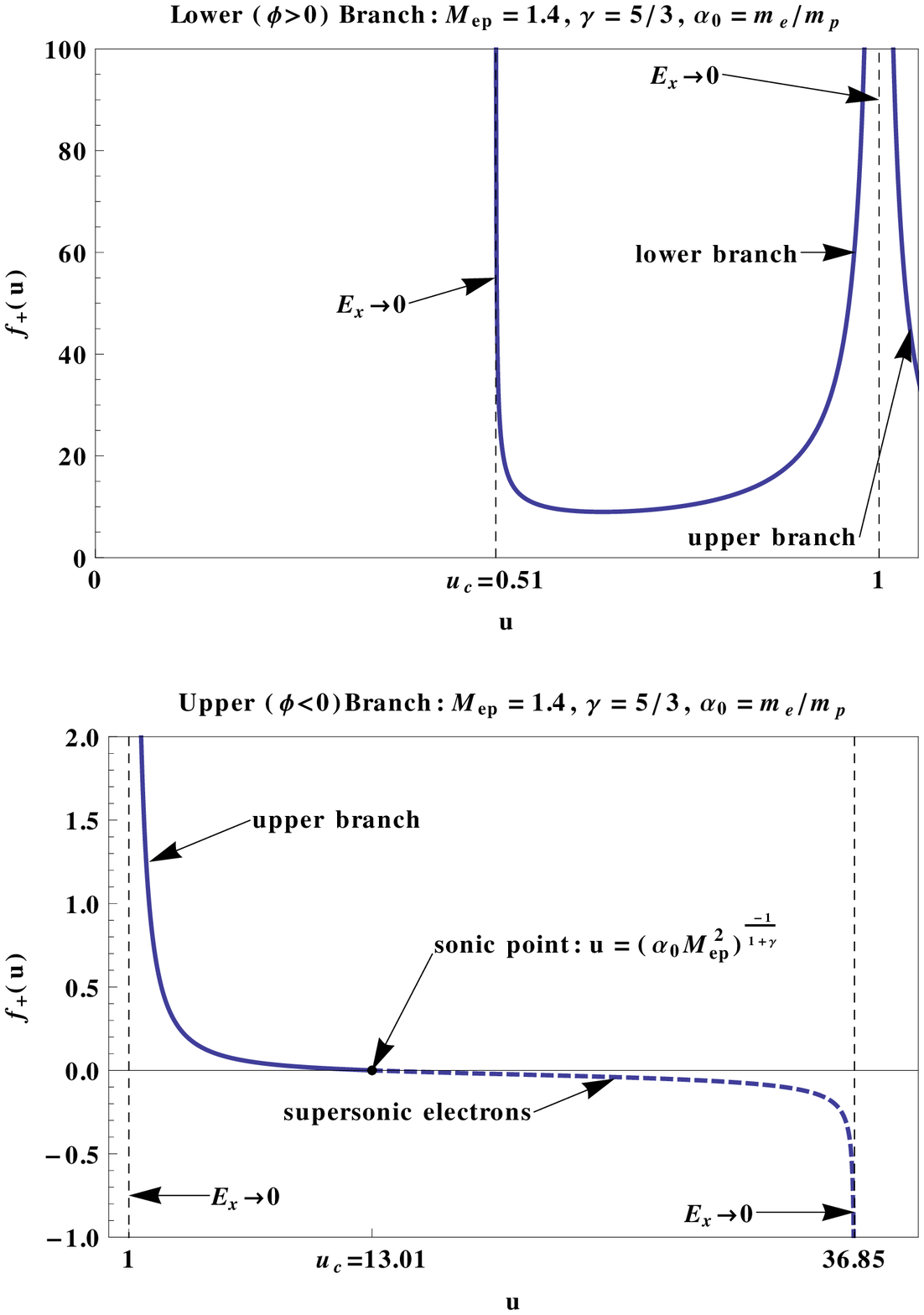}
\end{center}
\caption{The kernel of the structure equation (\ref{eqn:dimintsol}) plotted as a function of electron velocity (where the plus-sign has been selected). The vertical separatrix at $x=1$ separate the solution space into lower (physical) and upper (unphysical) branches.}\label{fig:fofu}
\end{figure}

\begin{figure}
\begin{center}
\includegraphics[width=0.75\textwidth, height=0.75\textheight]{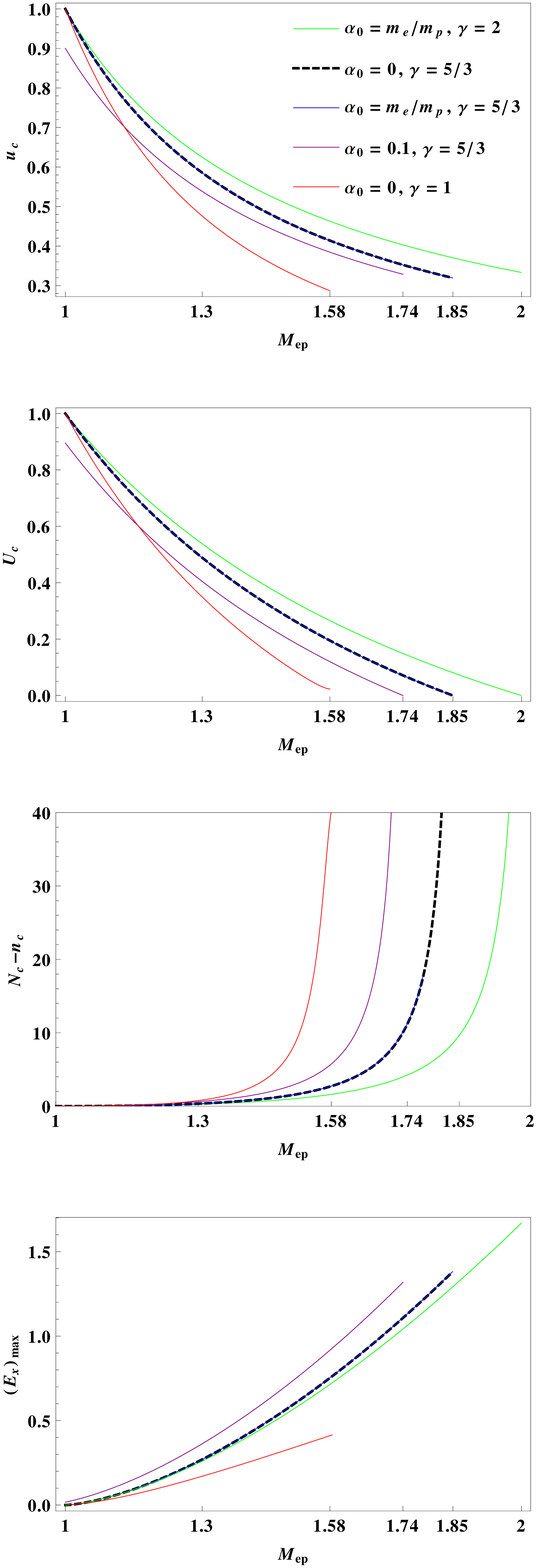}
\end{center}
\caption{Wave amplitude as a function of Mach number $M_{ep}$, 
determined by setting $E_x=0$ (at the center of the wave) in equations (\ref{eqn:ExsqofP}) through (\ref{eqn:Egofu}) to obtain the maximum values of $U$ and $u$, and by noting that the electric field reaches a maximum at the charge neutral points where $u=U$.}\label{fig:SolEnvlp}
\end{figure}

\begin{figure}
\includegraphics[width=1.0\textwidth, height=0.75\textheight]{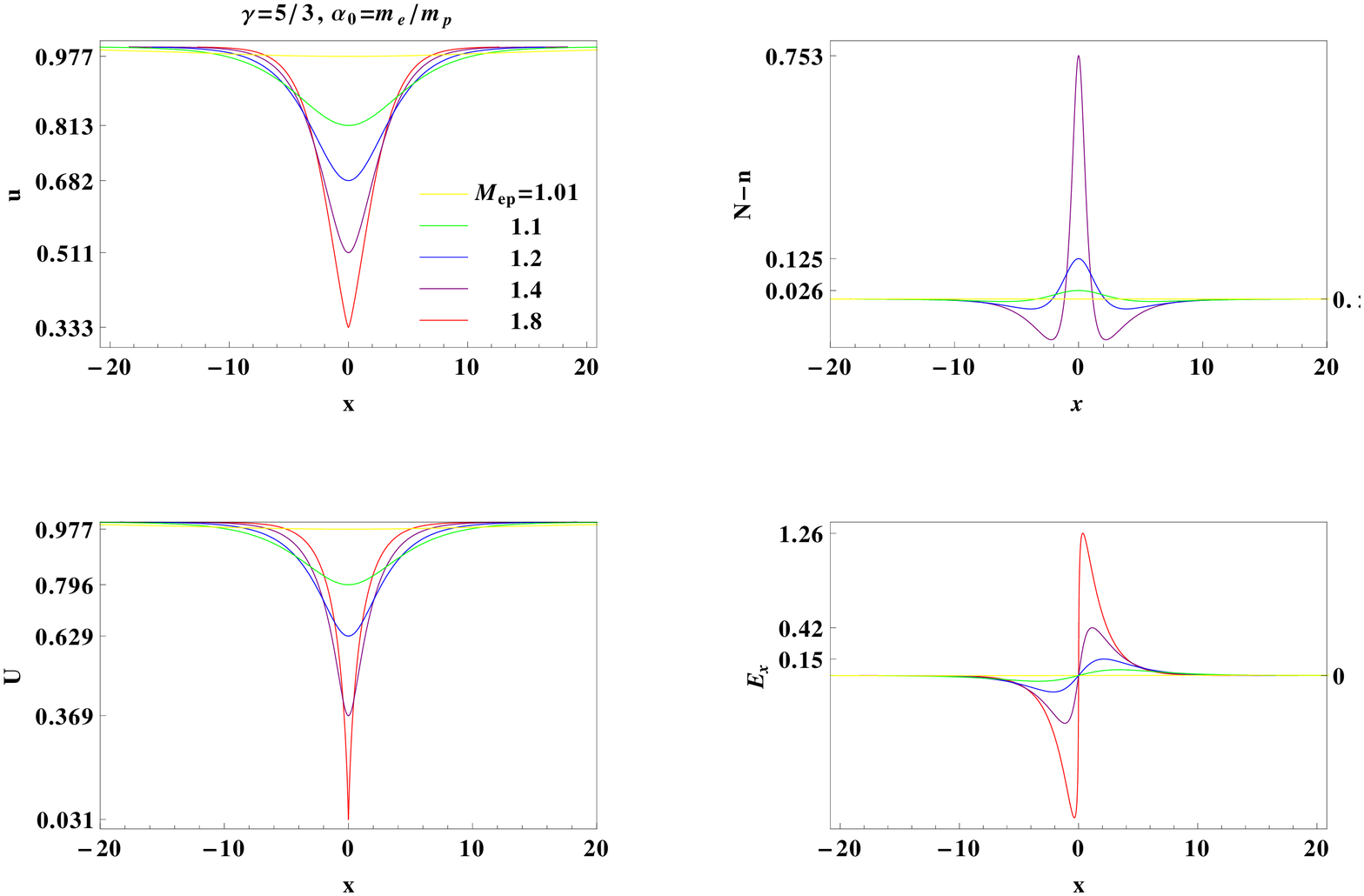}
\caption{Ion-acoustic solitary wave signatures for 
$\gamma\equiv\gamma_e = 5/3$ and $\alpha_0 = 0.0005446$, and Mach numbers 
$M_{ep}$ ranging from 1.01 to 1.8.}\label{fig:sltrywvs}
\end{figure}

\begin{figure}
\begin{center}
\includegraphics[width=0.45\textwidth, height=0.275\textheight]{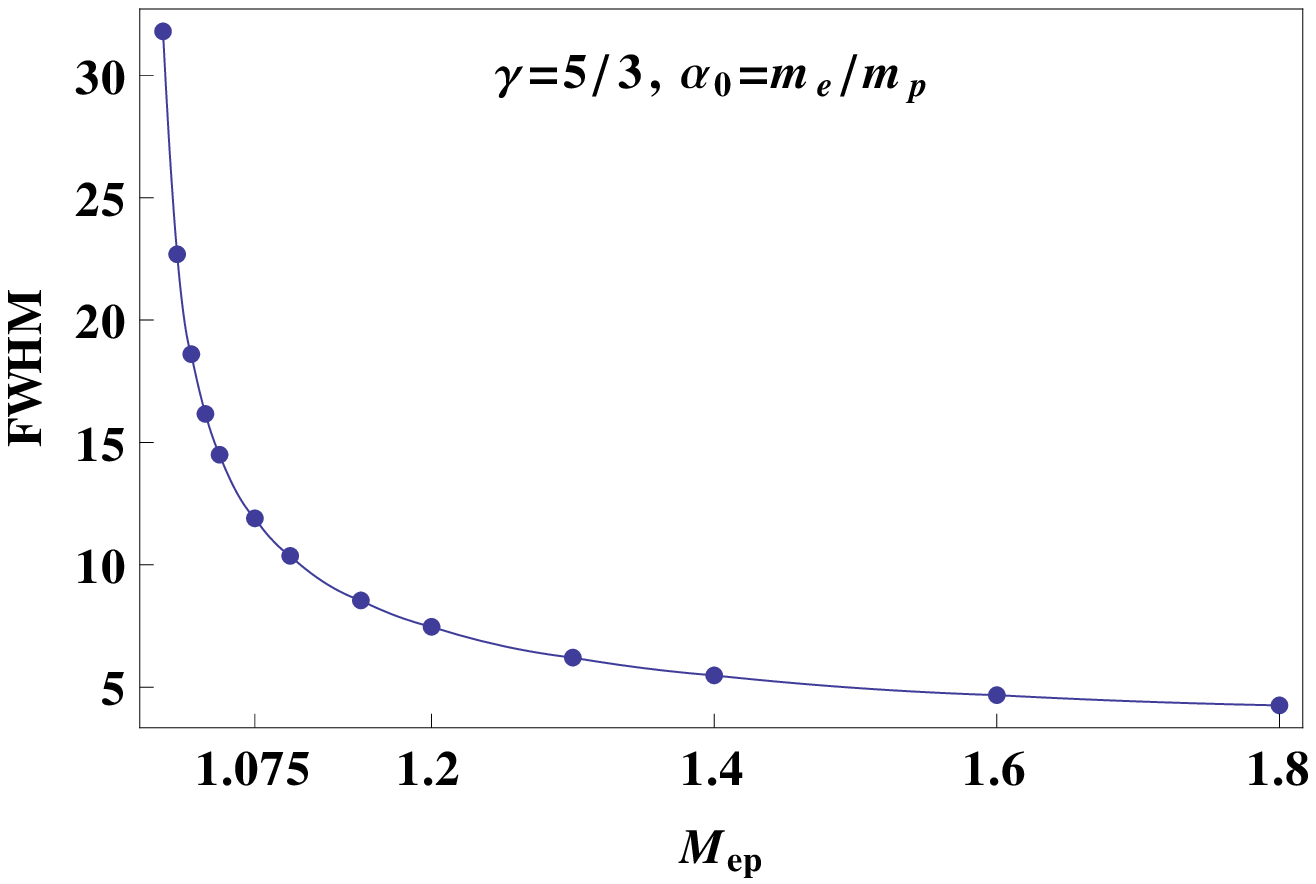}
\end{center}
\caption{Full width at half minimum (FWHM) of the electron velocity for solitary waves as characterized by Figure \ref{fig:sgntrM1p4}. 
On choosing a Mach number, $M_{ep}$,  arbitrary soliton structures can be hand sketched, by finding the corresponding wave amplitudes from Figure \ref{fig:SolEnvlp}, using Figure \ref{fig:sgntrM1p4} to note the typical wave structure, and finding the width of the wave using this figure.}\label{fig:fwhm}
\end{figure}

\section{Conclusion}
 In this paper we  have investigated the physical characteristics of ion-acoustic and Langmuir-acoustic travelling waves in a multi-fluid plasma. We first determined the dispersion equation for the system (section 3), which shows that there are two main branches, namely the ion-acoustic and the Langmuir wave 
branch. The long wavelength dispersion equation for the ion-acoustic branch 
results in the linearized KdV equation, whereas the long 
wavelength limit of the 
Langmuir wave branch results in a Klein-Gordon equation for linear waves. 
At short wavelengths, the basic wave modes are the electron acoustic 
and ion acoustic waves. These results are well known. Our main aim was to 
describe how the linear dispersion equation contains important 
information describing the travelling waves. 
In particular, far upstream the travelling wave 
is essentially a linear wave and the intersection of the dispersion equation 
$D(\omega,k)=0$ with the travelling wave condition $\omega=kV$ ($V$ 
is the travelling wave speed), yields the wave number of the wave far upstream. 

The main results of the paper are presented in Sections 4 and 5. 
In Section 4, we establish the multi-symplectic Hamiltonian structure 
for the travelling waves for an electron-ion two fluid plasma. We expect 
that the same multi-symplectic structure will also apply for multi-fluid 
plasmas that have a Hamiltonian structure governed by a canonical or 
non-canonical Poisson bracket (e.g. \cite{Spencer82}, 
\cite{SpencerKaufman82}, 
\cite{Holm83}). In the first formulation the Hamiltonian 
is identified as the total conserved longitudinal $x$-momentum of the system, 
in which the total energy flux integral acts as a constraint, 
and for which $d/d\xi$ is the Hamiltonian evolution operator ($\xi=x-Vt$ 
is the travelling wave variable and $V$ is the velocity of the wave). 
In this formulation, the canonical variables 
are $(\tilde{\varepsilon}_e,E_x)$ where 
\begin{equation}
\tilde{\varepsilon}_e=\frac{\varepsilon_0 m_e}{e}
\left(\frac{1}{2}u_e^2+\frac{c_e^2}{\gamma_e-1}\right), \label{eq:6.1}
\end{equation}
is the normalized electron energy flux and $E_x$ is the electric field 
intensity in the $x$-direction. 

In the second Hamiltonian formulation the total energy flux integral 
\begin{equation}
\varepsilon=j\left(\frac{1}{2} m_e u_e^2
+\frac{\gamma_e p_e}{(\gamma_e-1) n_e}\right)
+j\left(\frac{1}{2} m_p u_p^2
+\frac{\gamma_p p_p}{(\gamma_p-1) n_p}\right), \label{eq:6.2}
\end{equation}
is the Hamiltonian, and the total $x$-momentum integral 
\begin{equation}
P_x=j(m_e u_u+m_p u_p)+p_e+p_p-\frac{\varepsilon_0}{2} E_x^2, \label{eq:6.3}
\end{equation}
is held constant as a constraint. Here $j=n_e u_e=n_pu_p$ is the 
constant number density flux for both the protons and the electrons. In the 
latter Hamiltonian formulation, the Hamiltonian evolution operator is:
\begin{equation}
\frac{d}{d\tau}=u_p\frac{d}{d\xi}, \label{eq:6.4}
\end{equation}
where $\xi=x-Vt$ is the travelling wave variable. The canonical variables in 
this formulation are $\left(E_x,\tilde{\varepsilon}_e\right)$. 

Section 5 presented examples of solitary travelling wave solutions.
We used the \cite{McKenzie02} fluid dynamics approach to write the momentum 
equation in integral form. This form is especially convenient for calculating 
solitary wave signatures since, for a plasma system composed of any number 
and type of plasma species, it can be written as an integral over a single 
variable. By examining the special case of completely cold protons and 
hot (subsonic) electrons we have demonstrated how the fundamental integrals 
of the plasma system can be used to find the solitary wave existence conditions, the overall wave amplitude and the behavior (acceleration or deceleration) from the initial point of a given species.

The fluid dynamics expressions of conservation of energy and momentum, were used to describe ion-acoustic solitons, including 1) the range of Mach numbers over which ion-acoustic solitons can exist as a function of the electron to proton mass ratio $\alpha_0$ for several values of adiabatic gas index, showing that the range is slightly left-shifted and reduced for increasing $\alpha_0$ and the upper bound increases with increasing $\gamma$, 2) the fluid velocity behavior for each species, acceleration or deceleration from the initial point as summarized in Table \ref{tab:acdc}, and 3) the wave amplitude, Figure \ref{fig:SolEnvlp}, which indicates a very large charge imbalance at the wave center occurring for Mach numbers approaching the upper bound $M_{max}$.  The structure equation in integral form was employed to calculate solitary wave forms for a range of Mach numbers yielding the characteristic wave structure (Figure \ref{fig:sgntrM1p4}) and scale width--expressed in terms of FWHM (Figure \ref{fig:fwhm}). Our approach yields a rather complete description of ion-acoustic solitary waves for an arbitrary proton-electron mass ratio, which can be readily employed to calculate solitary wave structures to a high degree of numerical accuracy.

\begin{acknowledgments}
The work of RHB and XA was supported in part by grant  
AFOSR-FA9550-10-1-0084. GPZ was supported in part by 
NASA grants NN05GG83G and NSF grant
nos. ATM-03-17509 and ATM-04-28880. GPZ was also supported in part by 
NASA PRIME Grant NNG05EC85C with subcontract number A99132BT.

\end{acknowledgments}

\end{document}